\documentclass{article}
\usepackage{spconf,amsmath,graphicx, amsthm}
\usepackage{graphics,graphicx,epsfig}
\usepackage{rotating,subfig}
\usepackage{caption}
\usepackage[noend]{algpseudocode}
\usepackage{algorithmicx,algorithm}
\usepackage{balance}
\newtheorem{theorem}{Theorem}

\title{Optimal Segmented Linear Regression for Financial Time Series Segmentation}
\name{Chi-Jen Wu, Wei-Sheng Zeng and Jan-Ming Ho}
\address{Institute of Information Science,\\ Academia Sinca\\ \{cjwu\}@arbor.ee.ntu.edu.tw, \{wszeng,hoho\}@iis.sinica.edu.tw}

\begin{document}
%
\maketitle
\begin{abstract}
Given a financial time series data, one of the most fundamental and interesting challenges is the need to learn the stock dynamics signals in a financial time series data. A good example is to represent the time series in line segments which is often used as a pre-processing step for learning marketing signal patterns in financial computing. In this paper, we focus on the problem of computing the optimal segmentations of such time series based on segmented linear regression models. The major contribution of this paper is to define the problem of Multi-Segment Linear Regression (MSLR) of computing the optimal segmentation of a financial time series, denoted as the MSLR problem, such that the global mean square error of segmented linear regression is minimized. We present an optimum algorithm with two-level dynamic programming (DP) design and show the optimality of OMSLR algorithm. The two-level DP design of OMSLR algorithm can mitigate the complexity for searching the best trading strategies in financial markets. It runs in O($kn^2$) time, where $n$ is the length of the time series sequence and $k$ is the number of non-overlapping segments that cover all data points.
\end{abstract}
\section{Introduction}
\label{sec:intro}
Learning investment or trading signals from financial market data is one of most fundamental and interesting research challenges in both academia and industry~\cite{Murphy1999,Akansu2016}. For example, the American hedge fund, Renaissance Technologies\footnote{https://rentec.com} has leveraged financial signal processing technologies in stock trading for a long time. However, financial time series is difficult to summarize or be represented due to its highly non-stationary nature \cite{AbuMostafa2004}. Given a financial time series data, an initial processing step of learning the signal patterns is often to represent the time series in line segments to alleviate data uncertainty and noise~\cite{Lavrenko2000}.

In a segmentation process, a time series is divided into $k$ non-overlapping segments, and each segment is represented by a model to describe data points in the segment. The segment representation is measured using error functions depending on requirements of their applications. Actually, time series segmentation is widely for dimensionality reduction purposes in economics, engineering and science~\cite{Bingham2006}.


Time series segmentation has been extensively discussed in different domains and various models, which has resulted in a large number of works \cite{lin2007experiencing,Esling2012}. In 1961, the first version of time series segmentation problem is reported in \cite{Bellman61} and a dynamic programming (DP) algorithm with time complexity O($kn^3$) is also described. Time series segmentation also arises in data mining applications. The article \cite{Keogh2004} gives a review on applications of segmentation methods in data mining. The methods are classified into three categories, including sliding windows, bottom-up and top-down methods. The experimental comparisons showed that the bottom-up method results in better performance than other methods.

In the past few years, a few algorithms~\cite{Shatkay1996,RosmanNIPS2014,AcharyaICML2016,TerziSDM2006} have been proposed to reduce the time complexity of time series segmentation problem. The objective is to simplify represent of large scales time series data. The Piecewise Linear Approximation (PLA)~\cite{Shatkay1996} is a widely used approach for the segmentation task. Acharya et al.,~\cite{AcharyaICML2016} presented near-linear time algorithms that achieve a significant improvement compared to the DP approach on large time series. Interested reader can refer to Esling and Agon~\cite{Esling2012} who present a survey on approximation segmentation of time series. To the best our knowledge, previous approaches have not addressed the problem of optimum segmentation of financial time series. Most of them discussed segmentation methods in terms of approximation representation~\cite{Shatkay1996}, on-line processing \cite{RosmanNIPS2014} and their time complexity~\cite{AcharyaICML2016}.

In this paper, we are interested in the open question~\cite{Keogh1998}, how to best choose $k$, the optimal number of segments used to represent a particular time series. For financial trading strategies, $k$ is a measure of number of times of changes in market trend. It is also an indicator of how many time to trade in the market while receiving a reasonable amount of trading profits. Instead of answering the open question directly, we will start with focusing on minimizing global square error for a given $k$, and also derive the optimal representation of each of the $k$ segments.

Firstly, we formulate the Multi-Segment Linear Regression (MSLR) problem and define the MSLR square error as the performance index. Then, we present the Optimal Multi-Segment Linear Regression (OMSLR) algorithm, the two-level DP approach for producing the globally optimal segmentation. Finally, we show the optimality of the proposed OMSLR algorithm. The time complexity of the OMSLR algorithm is O($kn^2$), where $n$ is the length of the time series and $k$ is the number of non-overlapping segments that cover all data points. To the best our knowledge, this paper is the first to investigate the global optimal segmentation problem in time series processing, especially for financial time series.


This paper is organized as follows. In Section~\ref{sec:formulation} we present the formulation of segmentations as an optimization problem, named MSLR problem. In Section~\ref{sec:OMSLR} we present the OMSLR algorithm. Some segmentation experiments are presented in Section~\ref{sec:experimental}, and the results are summarized in Section~\ref{sec:conclusion}.

\section{Formulation of Problem MSLR}
\label{sec:formulation}
A formal definition of the Multi-Segment Linear Regression (MSLR) problem is described in this section. Given a time series $X = \{x_1, x_2,\dots,x_n\}$ and an integer $k$, the objective to MSLR problem is to partition $X$ into $k$ contiguous and non-overlapping intervals, i.e., $[l_{i-1},l_{i})$ and $[l_{k-1},l_{k}]$, where $l_0=1, l_k=n, 1\leq l_i \leq n, 1 \leq i \leq k-1, l_i \in \mathbf{N}$, such that the multi-segment linear regression square error, $\psi^2(1, n | \phi_k(X_n))$, with respect to the $k$-segment partition $\phi_k(X_n)=\{1, l_1, \dots, l_k\}$ is minimized. Note that $\psi^2(1, n | \phi_k(X_n))$ is also denoted as the Global Mean Square Error of the multi-segment linear regression representation, or GMSE for short. In other words, we have
\begin{equation}\label{eqn:1}
  \begin{aligned}
  \psi^2&(X_n | \phi_k(X_n)) \\
  =&\sum_{i=1}^{k-1} \sigma^2 (l_{i-1},l_i-1) + \sigma^2 (l_{k-1},l_k) \\
  =&\psi^2(X_{l_{k-1}-1}| \phi_{k-1}(X_{l_{k-1}-1})) + \sigma^2 (l_{k-1},l_k),
\end{aligned}
\end{equation}
where $\sigma^2(i,j)=\sum_{m=1}^{j}(x_m - \mu(i,j,m))^2$ and $\mu(i,j,m)) = \beta_{ij}*m + \alpha_{ij}$, $i \leq m \leq j$ with $\alpha_{ij}$ and $\beta_{ij}$ being the linear regression parameters on the interval $[i,j]$ of the time series $X_n$, i.e., $X_{ij}=\{x_i, x_{i+1}, \dots, x_j \}$. Thus we have $\alpha_{ij}$, $\beta_{ij}$ and $\sigma^2(i,j)$ as follows.

\begin{equation}\label{eqn:2}
\begin{aligned}
  \alpha_{ij}
  &= 	\bar x_{ij} - \beta_{ij} * \bar t_{ij} \\
  \beta_{ij}
  &= \frac{\sum_{m=i}^{j}(x_m - \bar x_{ij})(m - \bar t_{ij})}{\sum_{m=i}^{j}(m - \bar t_{ij})^2} \\
  \sigma^2(i,j)
  &=\sum_{m=1}^{j}(x_m - \mu(i,j,m))^2\\
  &=\sum_{m=1}^{j}(x_m - \beta_{ij} * m  - \alpha_{ij})^2,
\end{aligned}
\end{equation}
where $\bar t_{i,j} = \frac{(i+j)}{2}$, $\bar x_{i,j} = \frac{\sum_{m=i}^{j} x_m}{j-i+1}$, and $\mu(i,j,m) = \beta_{ij} * m + \alpha_{ij}, i \leq m \leq j$.  It can be shown that the above equations can be rewritten into iterative forms such that can be computed in O($n^2$) time for all $1 \leq i \leq j \leq n$. Due to space limitations we skip the detailed derivations here.

\section{OMSLR algorithm}
\label{sec:OMSLR}
We present the OMSLR algorithm for the MSLR problem as follows. Given a time series $X = \{x_1, x_2,\dots,x_n\}$ and an integer $k$, the algorithm OMSLR iteratively segments the time series $X_j = \{x_1, x_2, \dots, x_j\}$, where $1 \leq j \leq n$, into $i$ segments, starting with $i=1$ to $i=k$. Since Equation \ref{eqn:1} is an iterative function, we design a DP algorithm to compute the matrix $M$, in which $M[i,j]=(\gamma_{i,j},\rho^2_{i,j})$, for $i=1 \rightarrow k$, as a representation of the best way of partitioning $X_j$ into $i$ segments $\forall~j, 1 \leq j \leq n$. Here, $\gamma_{i,j}, 1 \leq \gamma_{i,j} \leq j$ denotes the starting point of the last segment of $ \hat \phi_i(X_j)$ and the variable $\rho^2_{i,j}$ is the global mean square error of $i$-segment partition of $X_j$ based on $\hat \phi_i(X_j)$.

$\gamma_{i,j}$ and $\rho^2_{i,j}$ can be computed by the following equations.
\begin{equation}\label{eqn:3}
\begin{aligned}
  \gamma_{i,j} &= \arg \min_{(i-1)d < m \leq j-d}\{ \rho^2_{i-1,m} + \sigma^2 (m+1, j)\};\\
  \rho^2_{i,j} &= \rho^2_{i-1,\gamma_{i,j}-1} + \sigma^2 (\gamma_{i,j}, j).
\end{aligned}
\end{equation}

\begin{algorithm}[t]
  \caption{OMSLR($X, k$)}\label{algo:OMSLR}
  \hspace*{0.02in} {\bf Input:}\\
  \hspace*{0.18in} $X \gets$ a time series data set, $n \gets$ the length of $X$\\
  \hspace*{0.18in} $k \gets$ the number of segments\\
  \hspace*{0.02in} {\bf Initialize:}\\
  \hspace*{0.18in} $gmse \gets \emptyset$ /* the Global Mean Square Error */\\
  \hspace*{0.18in} $\gamma_{i,j} \gets 0$, $\rho_{i,j} \gets 0$\\
  \hspace*{0.18in} $\sigma^2 \gets$ the pre-computed matrix $M$
  \begin{algorithmic}[1]
    \State /* for 2-segment linear regression */
    \For{$j = 1,2, \dots n$}
      \State $\gamma_{1,j} = \arg \min_{(1 < m \leq j)}\{ \sigma^2(1,m) + \sigma^2 (m+1, j)\}$
      \State $\rho^2_{1,j} = \sigma^2(1, \gamma_{1,j}-1) + \sigma^2(\gamma_{1,j}, j)$
    \EndFor
    \State end for
    \State /* dynamic programming for computing $k \geq 2$ */
    \For{$i = 2, \dots k-1$}
      \For{$j = 1,2, \dots n$}
        \State $\gamma_{i,j} = \arg \min_{(1 < m \leq j)}\{ \rho^2_{i-1,m} + \sigma^2 (m+1, j)\}$
        \State $\rho^2_{i,j} = \rho^2_{i-1,\gamma_{i,j}-1} + \sigma^2(\gamma_{i,j}, j)$
      \EndFor
      \State end for
    \EndFor
    \State end for
    \State /* backtrack $\gamma$ for the pivots */
    \State $pivot\_set \gets \{\}$
    \State $p\_seg \gets k$
    \State $p\_cur \gets n$
    \While{$ p\_seg > 0$}
    \State $pivot \gets \gamma_{p\_seg, p\_cur}$
    \State push $pivot$ into $pivot\_set$
    \State $p\_cur \gets pivot-1$
    \State $p\_seg \gets p\_seg-1$
    \EndWhile
    \State end while
    \State $gmse \gets \rho^2_{k,n}$
  \end{algorithmic}
  \hspace*{0.02in} {\bf Output:} $pivot\_set$, $gmse$
\end{algorithm}

In Equation \ref{eqn:3}, the $d$ is a constant used to control the minimum size of a segmentation, the default value of $d$ is 2. With the matrix $M$ and segmentation index $\gamma_{i,j}$, we can backtrack an $i$-segment partition on $X_j$ denoted as $ \hat \phi_i(X_j) = \{1, \hat l_1^{(i,j)}, \dots, \hat l_i^{(i,j)} \}$, by the following equations.

\[ \hat l_m^{(i,j)} = \left\{ \begin{array}{ll}
  j, \mbox{$m=i$};\\
  \gamma_{m+1,(\hat l_{m+1}^{(i,j)}-1)}, \mbox{$1 \leq m \leq i-1$}.
\end{array} \right.
\]

In specific, the $k$-segment partition of $X_n$, denoted as $ \hat \phi_k(X_n) = \{1, \hat l_1, \dots, \hat l_k \}$, is computed as follows.

\[ \hat l_i = \left\{ \begin{array}{ll}
  n, \mbox{$i = k$};\\
  \gamma_{i+1,( \hat l_{i+1}-1)}, \mbox{$1 \leq i \leq k-1$},
\end{array} \right.
\]
where $ x_{\hat l_0} = x_1$ and $x_{\hat l_k} = x_n$.
The algorithm OMSLR, given as Algorithm \ref{algo:OMSLR}, provides an optimal solution of the MSLR problem. In the following, we are going to show that $\hat \phi_k(X_n)$ is an optimal solution of the MSLR problem.

\begin{theorem}\label{theorem:1}
  Given a time series $X_n = \{x_1, x_2,\dots,x_n\}$ and the number of segments $k$, the $i$-segment partition $\hat \phi_{i+1}(X_j), \forall~j, 1 \leq j \leq n, \forall~i, 1 \leq i \leq k$ as computed by Algorithm OMSLR is optimum.
\end{theorem}

\begin{proof}
  We give a sketch of the proof and prove Theorem \ref{theorem:1} by contradiction. We skip the case $k=1$, it is a natural linear regression. For the case $k=2$, it is obviously to see that $\forall~j, 1 \leq j \leq n$, $\hat \phi_{2}(X_j)$ is optimum, since Algorithm OMSLR enumerated all the feasible solutions.

  In the induction step, we assume that $\forall~j, 1 \leq j \leq n$, $\hat \phi_{i}(X_j)$ is optimum. To show that $\hat \phi_{i+1}(X_j)$ is also optimum, $\forall~j, 1 \leq j \leq n$, we assume that there exits an integer $\tau, 1 \leq \tau \leq n$, such that $\hat \phi_{i+1}(X_{\tau})$ is not optimum. Let $\phi_{i+1}^{*}(X_{\tau})=\{1, l^*_1, \dots, l^*_i, l^*_{i+1}=\tau\}$ be the optimum $(i+1)$-segment partition of $X_{\tau}$. Then we have the following equation:
  \begin{equation}\label{eqn:4}
    \begin{aligned}
      \psi^2&(X_{\tau} | \hat \phi_{i+1}(X_{\tau})) > \psi^2(X_{\tau} | \phi_{i+1}^*(X_{\tau})).
    \end{aligned}
  \end{equation}
The induction assumption says that the $i$-segment partition, $\hat \phi_{i}(X_{l^*_{i}-1)}$ is optimum, thus it implies the Equation \ref{eqn:5}:
  \begin{equation}\label{eqn:5}
    \begin{aligned}
      \psi^2&(X_{l^*_{i}-1} | \phi_{i}^*(X_{l^*_{i}-1})) = \psi^2(X_{l^*_{i}-1} | \hat \phi_{i}(X_{l^*_{i}-1})),
    \end{aligned}
  \end{equation}
  and Algorithm OMSLR also guarantees the Equation \ref{eqn:6}.
  \begin{equation}\label{eqn:6}
    \begin{aligned}
      \psi^2(X_{\tau} | \hat \phi_{i+1}(X_{\tau})) &\leq \psi^2(X_{l^*_{i}-1} | \hat \phi_{i}(X_{l^*_{i}-1})) + \sigma^2 (l^*_{i}, \tau) \\
      &= \psi^2(X_{l^*_{i}-1} | \phi_{i}^*(X_{l^*_{i}-1})) + \sigma^2 (l^*_{i}, \tau) \\
      &= \psi^2(X_{\tau} | \phi^*_{i+1}(X_{\tau}))
    \end{aligned}
  \end{equation}
  Equation \ref{eqn:5} and \ref{eqn:6} imply that the assumption of Equation \ref{eqn:4} is a contradiction. Thus we have Theorem \ref{theorem:1}.
\end{proof}

\begin{figure}[t]
  \centering
  \centerline{\includegraphics[width=7.5cm]{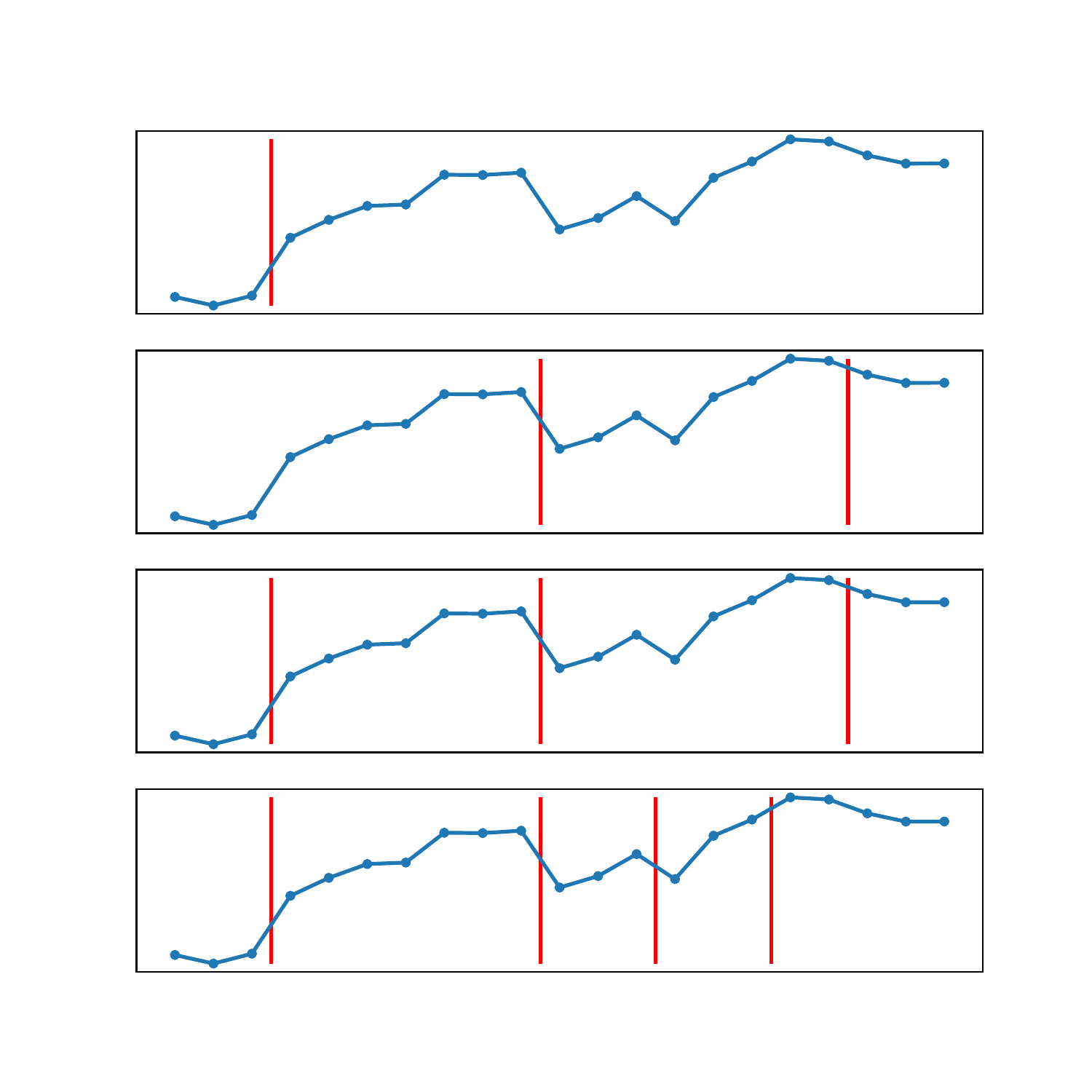}}
  \caption{A step by step results of OMSLR, for $k=2 \rightarrow 5$.}
  \label{fig:res1}
\end{figure}

The running time of the algorithm OMSLR is obviously to seen to take O($kn^2$) time. Due to the space limitation, we omit the detailed proof of Theorem \ref{theorem:2}.

\begin{theorem}\label{theorem:2}
  The running time of the algorithm OMSLR is at most O($kn^2$) for $X_n = \{x_1, x_2,\dots,x_n\}$ and $k$ is the number of non-overlapping segments of $X_n$.
\end{theorem}

\begin{figure}[t]
\begin{minipage}[b]{1.0\linewidth}
    \centering
    \centerline{\includegraphics[width=9cm]{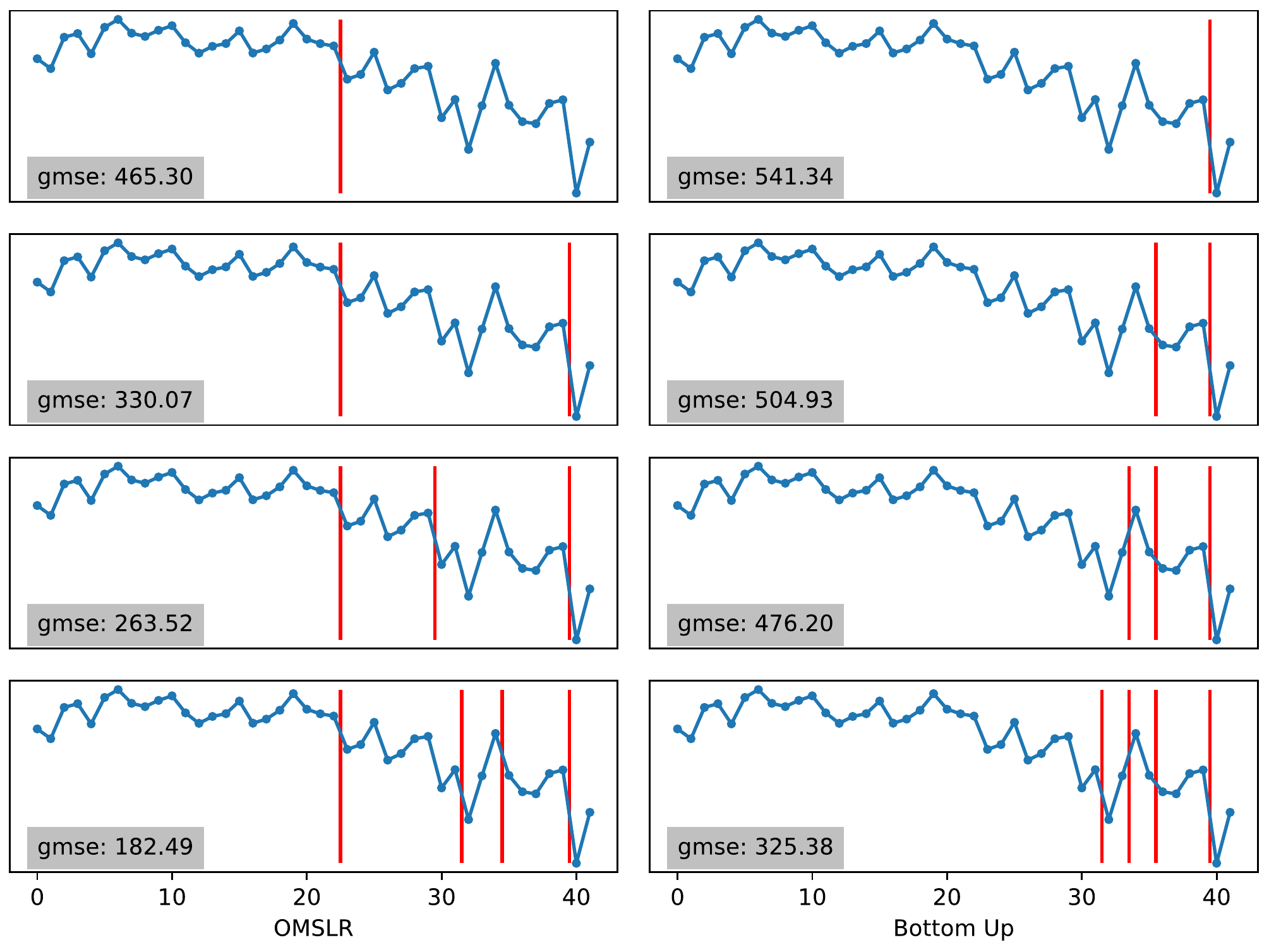}}
    \centerline{(a) The segmentations of OMSLR and Bottom-up}\medskip
  \end{minipage}
  \begin{minipage}[b]{1.0\linewidth}
    \centering
    \centerline{\includegraphics[width=8cm]{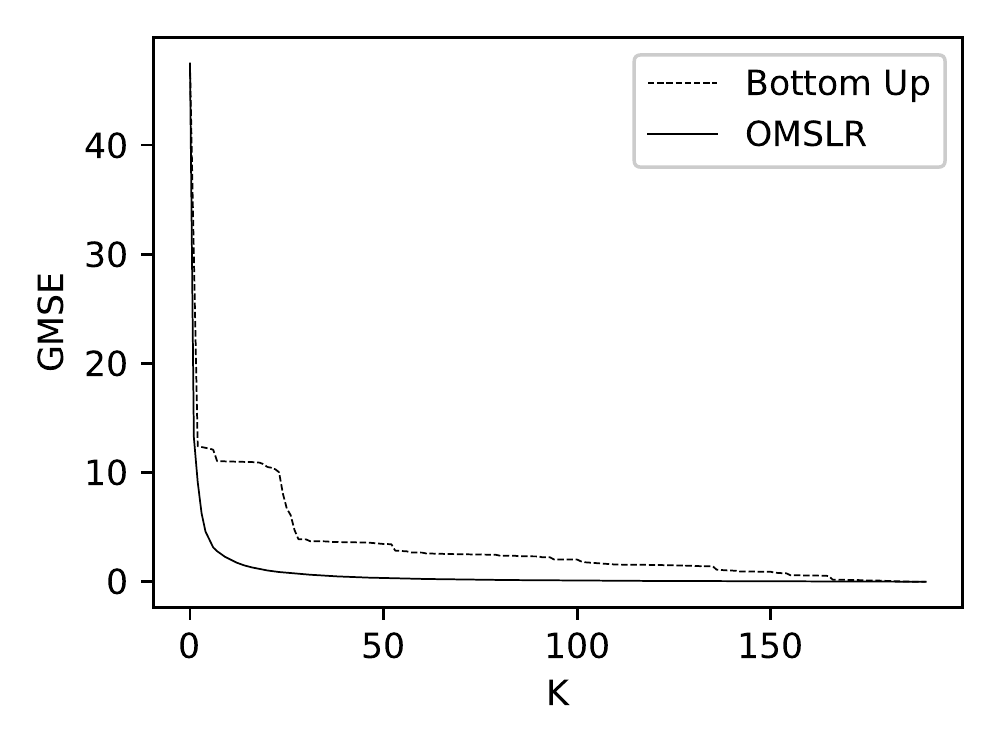}}
    \centerline{(b) The comparison of GMSE }\medskip
  \end{minipage}
  \caption{The experimental results.}
  \label{fig:res2}
\end{figure}

To illustrate how the algorithm OMSLR handles $k$-segment partition of $X_{n}$, we show a step-by-step results in Fig. \ref{fig:res1}. To keep things simple we assume that $k=5$ and $n=22$ in this example. In the Fig.~\ref{fig:res1}, each step of the $i$-segment partition is demonstrated, for $i=2 \rightarrow 5$. As shown as Theorem \ref{theorem:1}, each $i$-segment partition is an optimal result. As we presented in Section~\ref{sec:formulation} and Section~\ref{sec:OMSLR}, the first step of OMSLR algorithm is to generate the matrix $M$, which can be processed iteratively based on Equation \ref{eqn:2} in a DP way. So that the algorithm OMSLR, based on backtracking the matrix $M$, can derive each optimal $i$-segment partition by a DP approach as shown in Algorithm \ref{algo:OMSLR} (in Line:7 $\rightarrow$ Line:12). Algorithm \ref{algo:OMSLR} observably is a two-level DP design. Leveraging the two-level DP design, the algorithm OMSLR can return any $i$-segment partition of $X_{n}$ with a reasonable size of $i$ without having to re-compute from scratch. Due to its low complexity, the algorithm OMSLR offers an opportunity for one to search the best trading strategies in financial computing.

\section{Experimental Results}
\label{sec:experimental}
We provide an experimental evaluation of the two algorithms, i.e., our OMSLR and the Bottom-up algorithm \cite{Keogh2004}, for examining the performance in terms of Global Mean Square Error (GMSE), is defined in Section~\ref{sec:formulation}, with respect to the value of $k$. The two results are summarized in Fig \ref{fig:res2}.

In the first experiment, we compare the step-by-step segment partition as $k$ varies from 2 to 5 in Fig.~\ref{fig:res2} (a). For illustration purposes, we plot the time series with small sample size. The data only contains 42 data points, and spans a period from 2008-08-01 to 2008-09-30 selected from S\&P 500 index historical daily price data. In Fig.~\ref{fig:res2} (a), it is shown that OMSLR has smaller GMSE than the Bottom-up algorithm. It also shows that OMSLR always maintains optimality in partitioning the time series into multi-segment linear representation for each value of $k$.

In the second experiment, we focus on analyzing the relationship between $k$ and GMSE with a large sample size. We use S\&P 500 index historical 1-minute price data from 2010-07-01 to 2010-07-07 with a total of 1,560 data points. We compare the GMSE calculated by OMSLR and Bottom-up algorithm for each $k$ from 1 to 200. Fig.~\ref{fig:res2} (b) demonstrates that GMSE generated by the two algorithms both decreases monotonically, and sharply at the beginning. Therefore, a searching method can be designed for locating the best value of $k$ with a given GMSE bound since the curve is a monotonically decreasing function. Compared to Bottom-up algorithm, a much smaller number of segments is required for algorithm OMSLR to find a multi-segment linear regression representation of the given time series to satisfy a given GMSE bound.

\section{Conclusion and Future work}
\label{sec:conclusion}
In this paper we study the problem of optimal segmentation of financial time series based on segmented linear regression models. We present the OMSLR algorithm based on the two-level DP design. We show that the algorithm is optimum with time complexity O($kn^2$). We also demonstrate its application in analyzing financial time series. The representation generated by the algorithm OMSLR may be fed into other intelligent applications, e.g., to predict future trend of a financial market. The algorithm may also find further applications, e.g., we may use it as a benchmark for other on-line stock trading algorithms~\cite{Conegundes2020}. The on-line version of the OMSLR algorithm can also be used in stock trading. We may also use the algorithm in processing data of other application domains, such as medical and data science applications.

\vfill\pagebreak
\balance

\bibliographystyle{IEEEbib}
\bibliography{main}

\end{document}